\documentclass[10pt]{article}
\usepackage{amsmath, amsthm, amssymb, amsfonts, mathrsfs, mathtools, upgreek}
\usepackage{graphicx}
\usepackage{makeidx}
\usepackage{tikz-network}
\usepackage[left=2cm,right=2cm,top=2cm,bottom=2cm]{geometry}
\usepackage{caption}
\usepackage{subcaption}
\usepackage[section]{placeins}
\usepackage{enumitem}
\usepackage{tikz}
\usepackage{stmaryrd}
\usepackage{upquote}
\usepackage{authblk}

\usepackage{multicol}

\usepackage{hyperref}
\hypersetup{colorlinks=true, citecolor={blue}, linkcolor=blue, filecolor=magenta, urlcolor=cyan}
\usetikzlibrary{arrows.meta, positioning, decorations.pathreplacing}
\usepackage{cleveref}

\newtheorem{theorem}{Theorem}[section]

\newtheorem{lemma}[theorem]{Lemma}

\newtheorem{example}[theorem]{Example}
\newtheorem{corollary}[theorem]{Corollary}
\theoremstyle{definition}
\newtheorem{defi}[theorem]{Definition}

\numberwithin{equation}{section}

\newcommand{\divides}{\mid}

\newcommand{\Irr}{\operatorname{Irr}}

\newcommand{\aut}{\operatorname{Aut}}

\newcommand{\cay}{\operatorname{Cay}}
\newcommand{\real}{\operatorname{Re}}

\newcommand{\iu}{\mathbf{i}}

\newcommand{\spec}{\operatorname{Spec}}
\newcommand{\lcm}{\operatorname{lcm}}

\newcommand{\Nl}{\mathbb{N}}
\newcommand{\Zl}{\mathbb{Z}}
\newcommand{\Rl}{\mathbb{R}}
\newcommand{\Ql}{\mathbb{Q}}

\newcommand{\Cl}{\mathbb{C}}

\newcommand{\eu}{\mathbf{e}_u}
\newcommand{\ev}{\mathbf{e}_v}

\newcommand{\ld}{\lambda}

\title{Perfect state transfer in Grover walks on normal Cayley graphs}

\author{ Koushik Bhakta and Bikash Bhattacharjya\\
	Department of Mathematics\\
	Indian Institute of Technology Guwahati, India\\
	b.koushik@iitg.ac.in, b.bikash@iitg.ac.in }

\date{}
\begin{document}
	\maketitle
	
	\vspace{-0.3in}
	
	\begin{center}{\textbf{Abstract}}\end{center}
	\noindent 
	A Cayley graph $\cay(\Gamma,S)$ over a finite group $\Gamma$ is said to be normal if its connection set $S$ is a union of some conjugacy classes of $\Gamma$. This paper investigates perfect state transfer in Grover walks  on normal Cayley graphs. The Grover walk is a widely studied discrete-time quantum walk. We establish a necessary and sufficient condition for the occurrence of perfect state transfer on normal Cayley graphs. As applications, we derive explicit spectral criteria for perfect state transfer on Cayley graphs over  abelian groups, dicyclic groups,  and dihedral groups. These results yield several infinite families of Cayley graphs exhibiting perfect state transfer. We further obtain simple combinatorial characterizations of the existence of perfect state transfer on Cayley graphs over dihedral and dicyclic groups.  Our general characterization also recovers a number of previously known results as special cases. As a further consequence, we obtain a complete characterization of perfect state transfer on unitary Cayley graphs. In particular, we prove that exactly four graphs in the class of unitary Cayley graphs exhibit perfect state transfer.

	\vspace*{0.3cm}
	\noindent 
	\textbf{Keywords.} Grover walk, perfect state transfer, normal Cayley graph, abelian group, dicyclic and dihedral group\\
	\textbf{Mathematics Subject Classifications:} 05C50, 05C25, 81Q99
	
	\section{Introduction}
	
	The study of quantum walks on graphs is an important area that lies at the intersection of quantum computing and graph theory. It serves as a fundamental building block for various quantum algorithms and has applications in fields such as computer science and physics (see \cite{bose_pst, christandl}). There are two types of quantum walks: continuous-time quantum walk~\cite{periodic} and discrete-time quantum walk~\cite{godsil_dqw}. A discrete-time quantum walk is a quantum analogue of the classical random walk, where a particle moves on a graph in discrete steps according to specific rules. In a classical random walk, a particle moves from one vertex to another in a graph with a probability determined by the edges connecting that vertex. However, in a quantum walk, the movement of the particle is governed by the principles of quantum mechanics, allowing it to exhibit unique behaviors such as superposition and interference. Quantum walks have been extensively studied in both physics and mathematics. One of the most important phenomena of quantum walks is their capability for perfect state transfer, which essentially involves transferring between specific quantum states with probability $1$.  Perfect state transfer has been extensively studied in the setting of continuous-time quantum walks; see \cite{coutinho_asso, hermie_pair_state, pal_subdivided, singh} and the references therein. Although perfect state transfer in discrete-time quantum walks has been studied less, it has attracted significant attention; see \cite{barr_pst, kendon_pst}. In addition to perfect state transfer, several other phenomena in discrete-time quantum walks have been investigated; see, for instance, \cite{banerjee2, pgst1, kubota_convergence, ito, sarkar, zhan3, uniform_dqw}. 
	
	In this paper, we focus on the Grover walk, one of the most widely studied models of discrete-time quantum walks. Zhan~\cite{zhan1} constructed infinite families of circulant graphs that exhibit perfect state transfer in Grover walks. Later, using Chebyshev polynomials, Kubota and Segawa~\cite{kubota1} characterized complete multipartite graphs exhibiting perfect state transfer in Grover walks. More recently, Kubota and Yoshino~\cite{kubota2} completely characterized all circulant graphs of valency at most four that exhibit perfect state transfer in Grover walks.

	Normal Cayley graphs are an important class of graphs that combine high symmetry with a rich algebraic structure. Their connection sets are unions of conjugacy classes, which allow the adjacency matrix to be analyzed using the representation theory of finite groups. As a result, the spectral decomposition can often be determined explicitly. This makes normal Cayley graphs a natural extension of the well-understood abelian Cayley graphs to a broader class of nonabelian Cayley graphs while still allowing explicit spectral analysis. Since perfect state transfer is fundamentally a spectral phenomenon, normal Cayley graphs provide a natural and effective setting for studying quantum state transfer.

	
	The main contribution of this paper is a characterization of perfect state transfer in Grover walks on normal Cayley graphs. More precisely, in Theorem~\ref{pst_nomralcayley_main}, we establish necessary and sufficient conditions for the occurrence of perfect state transfer on normal Cayley graphs over finite groups. We then apply this characterization to several important classes of finite groups. In particular, Theorems~\ref{pst_abelian}, \ref{pst_dicyclic}, and \ref{pst_dihedral} characterize perfect state transfer on normal Cayley graphs over abelian, dicyclic, and dihedral groups, respectively. As applications of these results, we construct several infinite families of normal Cayley graphs that exhibit perfect state transfer (Lemma~\ref{eg_pst_abelian}, Corollary~\ref{eg_coro_pst_abelian}, Examples~\ref{eg_circulant}, \ref{eg_dicyclic}, and \ref{eg_dihedral}). Furthermore, for dicyclic and dihedral groups, we derive a combinatorial criterion for perfect state transfer in Corollaries~\ref{comb_dicyclic} and \ref{comb_dihedral}, respectively. Finally, as another application of our general characterization, we obtain a complete characterization of perfect state transfer on unitary Cayley graphs in Theorem~\ref{pst_unitary}.

	The rest of the paper is organized as follows. Section~2 provides the necessary preliminaries on Grover walks and the representation theory of finite groups. In Section~3, we establish a characterization of perfect state transfer on normal Cayley graphs and then apply it to abelian, dicyclic, and dihedral groups. In Section~4, we apply the general characterization to unitary Cayley graphs.

	\section{Preliminaries}\label{prel}

	Let $G:=(V(G),E(G))$ be a finite simple graph with vertex set $V(G)$ and edge set $E(G)$.  We write the elements of $E(G)$ as $uv$, where $u,v\in V$, $u\neq v$ and, $uv$ and $vu$ represent the same edge. Let $A:=A(G)\in \Cl^{V(G)\times V(G)}$ be the \emph{adjacency matrix} of $G$, where
	$$A_{u,v} = \left\{ \begin{array}{rl}
		1 &\mbox{ if }
		uv\in E(G) \\ 
		0 &\textnormal{ otherwise.}
	\end{array}\right.$$  
	A graph $G$ is said to be \emph{integral} if all the eigenvalues of its adjacency matrix are integers. For a matrix $M$ associated to a graph $G$, let $\spec_M(G)$ denote the set of all distinct eigenvalues of $M$.

	\subsection{Grover walks}
	If $uv$ is an edge of a graph $G$, then the ordered pairs $(u,v)$ and $(v,u)$ are called the \emph{arcs} of $G$ associated to the edge $uv$. We define $\mathcal{A}(G):=\{(u, v), (v, u):uv\in E(G) \}$, the set of all symmetric arcs of $G$. The vertices $u$ and $v$ are called the \emph{origin} and \emph{terminus} of an arc $(u,v)$, respectively. Let $a$ be an arc of $G$, where $a=(u,v)$. We write $o(a)$ and $t(a)$ to denote the origin and terminus of $a$, respectively, that is, $o(a)=u$ and $t(a)=v$.  The inverse arc of $a$, denoted $a^{-1}$, is the arc  $(v,u)$.
	
	We now define a few matrices required for the definition of Grover walks. The \emph{boundary matrix} $N:=N(G)\in \mathbb{C}^{V(G)\times \mathcal{A}(G)}$ of $G$ is defined by $N_{u,a}=(1/\sqrt{\deg u})\delta _{u, t(a)},$ where $\delta_{a,b}$ is the Kronecker delta function. The \emph{shift matrix} $R:=R(G)\in \mathbb{C}^{\mathcal{A}(G) \times \mathcal{A}(G)}$ of $G$ is defined by $R_{a,b}=\delta_{a,b^{-1}}$. Define the \emph{time evolution matrix} $U:=U(G)\in \mathbb{C}^{\mathcal{A}(G)\times \mathcal{A}(G)}$ of $G$ by $$U=R(2N^*N-I).$$ The time evolution matrix is also known as the Grover transition matrix. A discrete-time quantum walk on a graph G is determined by a unitary matrix, which acts on the complex functions of the symmetric arcs of G. The discrete-time quantum walks defined by $U$ are called the \emph{Grover walks}. The Grover walks are also known as a special case of bipartite walks~\cite{chen}. The entries of the time evolution matrix are calculated as $$U_{a,b}=\frac{2}{\deg t(b)}\delta_{o(a),t(b)}-\delta_{a,b^{-1}}.$$ 
	The eigenvalues of the Grover transition matrix $U$ can be determined from the smaller-sized matrix, known as the \emph{discriminant} $P:=P(G)\in \mathbb{C}^{V(G)\times V(G)}$ of $G$, defined by $P=NRN^*$.
	See \cite{kubota2} for more details about the matrices $N$, $R$, $U$ and $P$. If  $G$ is a regular graph, then the matrix $P$ can be expressed in terms of its adjacency matrix $A$.
	\begin{lemma}[{\cite[Lemma~3.1]{bhakta4}}]\label{reg}
		If $G$ is a $k$-regular graph, then $P=\frac{1}{k}A$. Further, $1$ is an eigenvalue of $P$, and for any eigenvalue $\mu$ of $P$,  $|\mu|\leq1$.
	\end{lemma}
	The following result is known as the spectral mapping theorem of the Grover walks. For a nonnegative integer $k$, by $\{a\}^k$, we denote the multi-set $\{a,\ldots,a\}$, where the element $a$ repeats $k$ times. We also denote $\sqrt{-1}$ by $\iu$.
	\begin{lemma}[{\cite[Proposition~1]{hig3}}]\label{ev_grover}  
		Let $\mu_1, \hdots,\mu_n$ be the eigenvalues of the discriminant of a graph $G$. Then the multi-set of eigenvalues of the time evolution matrix is $$\left\{\exp\!\left(\pm \iu \arccos \mu_j\right): 1\leq j\leq n\}\right\}\cup \{1\}^{b_1} \cup \{ -1\}^{b_1-1+1_B},$$ where $b_1=|E(G)|-|V(G)|+1$, and $1_B=1$ or $0$ according as $G$ is bipartite or not.	
	\end{lemma}
	
	\begin{defi}
		The Grover walk on a graph is said to be \emph{periodic} if its transition matrix $U$ satisfies $U^\tau=I$ for some positive integer $\tau$.
	\end{defi}	
	For simplicity, we say that the graph $G$ is \emph{periodic} whenever the Grover walk on $G$ is periodic. If $\tau$ is the smallest positive integer such that $U^\tau=I$, then $\tau$ is called the \emph{period} of the graph, and the graph is said to be \emph{$\tau$-periodic}.

	Since the time evolution matrix $U$ is unitary, it is diagonalizable. Consequently, the periodicity of a graph can be characterized entirely in terms of the eigenvalues of $U$.
	\begin{lemma}[{\cite[Lemma~5.3]{mixed2}}] \label{period}
		A graph $G$ is periodic  if and only if $ \eta^\tau  =1$ for each $\eta \in \spec_U(G)$ for some positive integer $\tau$. Moreover, $G$ is $\tau$-periodic if and only if $\tau$ is the least positive integer such that $\eta^\tau=1$ for each $\eta \in \spec_U(G)$.
	\end{lemma}
	The preceding lemma shows that the period of a graph is completely determined by the orders of the eigenvalues of its time evolution matrix. This immediately yields the following corollary.
	\begin{lemma} \label{periodic}
		Let $\eta_1, \hdots , \eta_m$ be the eigenvalues of the time evolution matrix of a periodic graph $G$. Let $k_1, \hdots, k_m$ be the least positive integers such that $\eta_1 ^{k_1}=1, \hdots,\eta_m^{k_m}=1$. Then the period of $G$ is $\lcm(k_1, \hdots, k_m)$.
	\end{lemma}
	
	Let $\Delta=\{a\pm\sqrt{b}:a,b\in \Ql~\text{and}~b ~\text{is not a square}\}$, and $\overline{\Delta}=\Rl\setminus(\Ql\cup\Delta)$. For a subset $F$ of real numbers, let $\spec_P^F(G)=F\cap\spec_P(G)$.
	For a complex number $z$, we denote $\real(z)$ as the real part of $z$. Let $\Re$ be the set of the real parts of the roots of unity.
	\begin{lemma}[{\cite[Theorem~4.4]{bhakta2}}]\label{regular_periodic}
		Let $G$ be a regular graph with discriminant matrix $P$. Then $G$ is periodic if and only if $$\spec_P^\Ql(G)\subseteq\left\{\pm1,\pm\frac{1}{2},0\right\},~\spec_P^\Delta(G)\subseteq\left\{\pm \frac{\sqrt{3}}{2}, \pm\frac{1}{4}\pm\frac{\sqrt{5}}{4},\pm\frac{1}{\sqrt{2}}\right\}~\text{and}~\spec_P^{\overline{\Delta}}(G)\subset \Re.$$
	\end{lemma}

	Let $G$ be a graph and $U$ be its time evolution matrix. For $\Phi\in \Cl^{\mathcal{A}(G)}$, we denote by $\|\Phi\|$ the Euclidean norm of $\Phi$.  A vector $\Phi\in \Cl^{\mathcal{A}(G)}$ is said to be a \emph{state} if $\|\Phi\|=1$. We say that \emph{perfect state transfer} occurs from a state $\Phi$ to another state $\Psi$ at time $\tau\in \Nl$ if there exists a unimodular complex number $\gamma$ such that $U^\tau\Phi=\gamma \Psi$. 
	
	We focus on the transfer of states localized at individual vertices of a graph. A state $\Phi\in \Cl^{\mathcal{A}(G)}$ is called \emph{vertex-type} in $G$ if there exists a vertex $u\in V(G)$ such that $\Phi=N^* \eu$. We denote the vertex-type state corresponding to the vertex $u$ by $\Phi_u:=N^*\eu$, and let $\chi :=\{\Phi_u: u\in V(G)\}$ denote the set of all vertex-type states. Throughout this paper, we restrict our attention to vertex-type states. For the motivation behind vertex-type states, we refer the reader to \cite{kubota1}.
	\begin{defi}
		A graph exhibits \emph{perfect state transfer} from vertex $u$ to vertex $v$ at time $\tau$ if there exists a complex number $\gamma$ with $|\gamma|=1$ such that $U^\tau \Phi_u=\gamma \Phi_v$.
	\end{defi}
	We say that a graph $G$ exhibits {perfect state transfer} if there exist vertices $u$ and $v$ in $G$ and a positive integer $\tau$ such that perfect state transfer occurs from $u$ to $v$ at time $\tau$.	In Grover walks, the occurrence of perfect state transfer between vertex-type states is closely related to Chebyshev polynomials.

	The \emph{Chebyshev polynomial of the first kind}, denoted $T_n(x)$, is a polynomial defined by the recurrence relation
	\[T_0(x)=1,~T_1(x)=x,~\text{and}~T_n(x)=2xT_{n-1}(x)-T_{n-2}(x) ~\text{for}~ n \geq 2.\]
	It is well known that 
	\begin{equation}\label{chb}
		T_n(\cos\theta)=\cos(n\theta).
	\end{equation}
	This implies that $|T_n(x)|\leq 1$ for $|x|\leq 1$. Thus the next lemma follows easily.
	\begin{lemma}\label{ch}
		Let $\mu\in[-1,1]$ and $n$ be any positive integer. Then
		\begin{enumerate}[label=(\roman*)]
			\item $|T_n (\mu)|\leq 1$.
			\item $T_n (\mu)= 1$ if and only if $\mu =\cos\frac{s}{n}\pi$ for some even positive integer $s$.
			\item $T_n (\mu)= -1$ if and only if $\mu =\cos\frac{s}{n}\pi$ for some odd positive integer $s$.
		\end{enumerate}
	\end{lemma}

	Kubota and Yoshino in \cite[Theorem~6.5]{kubota2} proved that a graph exhibits perfect state transfer from vertex $u$ to vertex $v$ at time $\tau$ if and only if $T_\tau(P)\eu=\gamma\ev$ for some $\gamma\in\{-1,1\}$. Subsequently, Guo and Schmeits \cite{guo_sch} showed that $\gamma=1$ is necessary. For an explicit proof of this fact, we refer the reader to \cite[Lemma~3.3]{bhakta4}.
	\begin{lemma}[{\cite[Theorem~6.5]{kubota2}}]\label{pst_vertex}
		Let $u$ and $v$ be two distinct vertices of a graph $G$. Then $G$ exhibits perfect state transfer from $u$ to $v$ at time $\tau$ if and only if $T_\tau(P)\eu=\ev$.
	\end{lemma}
	Moreover, Kubota and Segawa \cite[Lemma~3.2]{kubota1} proved that $\|T_n(P)\eu\|\le1$ for each vertex $u$. Hence perfect state transfer occurs from $u$ to $v$ at time $\tau$ if and only if $|T_\tau(P)_{u,v}|=1$, that is, $T_\tau(P)_{u,v}=1$. Since $P$ is symmetric, perfect state transfer occurs from  $u$ to $v$  if and only if it occurs from $v$ to $u$. Hence we say that perfect state transfer occurs between $u$ and $v$.
	
	\subsection{Cayley graphs}	
	
	Let $\Gamma$ be a finite group and $S$ be an inverse closed subset of $\Gamma\setminus\{e\}$, where $e$ is the identity element of $\Gamma$. The \emph{Cayley graph} $\cay(\Gamma, S)$ of $\Gamma$ with respect to $S$ is an undirected graph whose vertex set is $\Gamma$ and two vertices $u$ and $v$ are adjacent  if $uv^{-1}\in S$. We say that $\cay(\Gamma,S)$ is \emph{normal} (or \emph{quasi-abelian}) if $S$ is the union of some conjugacy classes of $\Gamma$. 
	
	Let $\Gamma$ be a finite group. A \emph{representation} of $\Gamma$ is a group homomorphism $\rho:\Gamma\to\operatorname{GL}(W)$ for some finite-dimensional vector space $W$, where $\mathrm{GL}(W)$ denotes the group of all invertible linear operators on $W$.  The \emph{degree} of $\rho$, denoted $d_{\rho}$, is the dimension of $W$. The one-dimensional representation $\rho_0:\Gamma\to\operatorname{GL}(\mathbb{C})$ defined by $\rho_0(g)=1$ for all $g\in\Gamma$, 	is called the \emph{trivial representation} of $\Gamma$.

	Two representations $\rho_1:\Gamma\to \mathrm{GL}(W_1)$ and $\rho_2:\Gamma\to \mathrm{GL}(W_2)$ of $\Gamma$ are said to be \emph{equivalent}  if there exists a linear isomorphism $T:W_1\to W_2$ such that $\rho_2(g)=T\,\rho_1(g)\,T^{-1} \text{ for each }g\in\Gamma$.

	Let $\rho:\Gamma\to\mathrm{GL}(W)$ be a representation of $\Gamma$. The \emph{character} of $\rho$ is the function $\chi_\rho:\Gamma\to\mathbb{C}$ defined by $\chi_\rho(g)=\operatorname{tr}(\rho(g))
	\text{ for }g\in\Gamma$, where $\operatorname{tr}(\rho(g))$ denotes the trace of the linear transformation $\rho(g)$. The degree of the character $\chi_\rho$ is defined to be the degree of the corresponding representation $\rho$, and is denoted by the same symbol $d_\rho$.
	
	A subspace $X$ of $W$ is called \emph{$\Gamma$-invariant} if $\rho(g)x\in X\text{ for each }g\in\Gamma \text{ and }x\in X$.
	The subspaces $\{0\}$ and $W$ are called the \emph{trivial invariant subspaces}. If these are the only $\Gamma$-invariant subspaces of $W$, then $\rho$ is called an \emph{irreducible representation} of $\Gamma$, and its character $\chi_\rho$ is called an \emph{irreducible character} of $\Gamma$.
	
	\begin{lemma}[{\cite[Proposition~4]{rep_dihedral}}]\label{schur_lemma}
		Let $\rho:\Gamma\to\operatorname{GL}(W)$ be an irreducible representation of $\Gamma$. If $T$ is a linear operator on $W$ such that $T\rho(g)=\rho(g)T$ for each $g\in\Gamma$, then $T$ is a scalar multiple of the identity.
	\end{lemma}
	Let $W$ be an inner product space. A representation $\rho:\Gamma\to\mathrm{GL}(W)$ is called \emph{unitary} if $\rho(g)$ is a unitary operator for every $g\in\Gamma$. It is well known that every representation of a finite group is equivalent to a unitary representation. Let $\widehat{\Gamma}$ denote a complete set of pairwise inequivalent irreducible unitary representations of $\Gamma$, and let $\operatorname{Irr}(\Gamma):=\{\chi_\rho:\rho\in\widehat{\Gamma}\}$
	denote the set of irreducible characters of $\Gamma$. We refer to \cite{steinberg} for more details about the representation theory of finite groups.
	
	The following result gives the eigenvalues and corresponding eigenvectors of the adjacency matrix of a normal Cayley graph.
	\begin{theorem}[{\cite[Exercise~5.12.3]{steinberg}}]\label{ev_normal_cayley}
		Let $\Gamma$ be a finite group and let $S$ be an inverse closed subset of $\Gamma\setminus\{e\}$ satisfying $gSg^{-1}=S$. Then the eigenvalues of the adjacency matrix of the Cayley graph $\cay(\Gamma,S)$  are given by 
		\[\lambda_\rho=\frac{1}{d_\rho}\sum_{s\in S}\chi_\rho(s)~~\text{for }\rho\in\widehat{\Gamma},\]
		where the eigenvalue $\lambda_\rho$ has multiplicity $d_\rho^2$. Moreover, for each $\rho\in\widehat{\Gamma}$, the vectors
		\[v_{i,j}^\rho=\frac{\sqrt{d_\rho}}{\sqrt{|\Gamma|}}\big(\rho_{i,j}(g):g\in\Gamma\big)^t~~\text{for }1\leq i,j\leq d_\rho\]
		form an orthogonal basis for the eigenspace corresponding to $\lambda_\rho$, where $\rho_{i,j}(g)$ denotes the $(i,j)$-entry of the matrix representing $\rho(g)$ with respect to a fixed orthonormal basis.
	\end{theorem}

	\section{Perfect state transfer on normal Cayley graphs}
	Let $\Gamma$ be a finite group. For each $g \in \Gamma$, define two permutations of $\Gamma$, namely the \emph{left multiplication} $\mathcal{L}_g:\Gamma\to\Gamma$ and the \emph{right multiplication} $\mathcal{R}_g:\Gamma\to\Gamma$, by $\mathcal{L}_g(h) = gh$ and $\mathcal{R}_g(h) = hg$, respectively, for all $h \in \Gamma$. Given a permutation $\sigma$ of $\Gamma$, its associated permutation matrix in $\mathbb{C}^{\Gamma\times\Gamma}$ has $(u,v)$-entry $\delta_{u,\sigma(v)}$.	Throughout the paper, we use the same notation, $\sigma$, for both a permutation and its associated permutation matrix whenever the context makes the meaning clear. Let $Z(\Gamma):=\{g\in\Gamma: gh=hg~\text{for each } h\in\Gamma\}$ be the center of the group $\Gamma$. For a graph $G$, let $\aut(G)$ be its automorphism group.
	\begin{lemma}\label{pst_normalcayley}
		Let $G:=\cay(\Gamma,S)$ be a normal Cayley graph over a finite group $\Gamma$ with connection set $S$. Then $G$ exhibits perfect state transfer at time $\tau$ if and only if  $T_\tau(P)=\mathcal{R}_z$ for some $z\in Z(\Gamma)$ of order $2$.
	\end{lemma}
	\begin{proof}
		Suppose $G$ exhibits perfect state transfer between the vertices $u$ and $v$ at time $\tau$. By Lemma~\ref{pst_vertex},  we have $ T_{\tau}(P)\eu=\ev$.  Using Lemma~\ref{reg}, $T_\tau(P)$ is a polynomial in $A$. Further, $\mathbf{e}^t_{xg}A^k\mathbf{e}_{yg} =\mathbf{e}^t_{x}A^k\mathbf{e}_{y}$ for any $x,y,g\in\Gamma$ and nonnegative integer $k$. Thus it follows that $\mathbf{e}^t_{xg}T_\tau(P)\mathbf{e}_{yg} =\mathbf{e}^t_{x}T_\tau(P)\mathbf{e}_{y}$. Combining this with the symmetry of $T_\tau(P)$ and $T_{\tau}(P)\eu=\ev$, we deduce that $T_\tau(P)$ is a permutation matrix of order $2$ with no fixed points.  
		
		Since $G$ is normal, both $\mathcal{L}_g$ and $\mathcal{R}_g$ are automorphisms of $G$ for each $g\in\Gamma$. Since $T_\tau(P)$ is a polynomial in $A$, it commutes with every element of $\aut(G)$. In particular, $T_\tau(P) \in Z(\aut(G)).$ Hence $T_\tau(P)\mathcal{L}_g = \mathcal{L}_g T_\tau(P)$ for all  $g\in \Gamma.$
		Thus for all $g,h \in \Gamma$, $T_\tau(P)(gh) = gT_\tau(P)(h).$ Let $z := T_\tau(P)(e)$, where $e$ is the identity element of $\Gamma$.	Setting $h=e$ in the previous identity yields $T_\tau(P)(g) = gz$ for all $g \in \Gamma$. 
		Therefore, $T_\tau(P) = \mathcal{R}_z.$ Since $T_\tau(P) \in Z(\aut(G))$, we must have $\mathcal{R}_z \mathcal{R}_g = \mathcal{R}_g \mathcal{R}_z$ for all $g \in \Gamma,$
		which holds precisely when $z \in Z(\Gamma)$.  
		Finally, since $T_\tau(P)$ is a permutation matrix of order $2$, $z$ must be an element of order~$2$. 
		
		Conversely, suppose that $T_\tau(P)=\mathcal{R}_z$ for some $z\in Z(\Gamma)$ of order $2$. Then for every $u\in\Gamma$, setting $v=uz$, we have $\mathbf{e}_v^{t}T_\tau(P)\mathbf{e}_u=1$. Therefore by Lemma~\ref{pst_vertex}, $G$ exhibits perfect state transfer between $u$ and $v$ at time $\tau$. This completes the proof.
	\end{proof}

	The following result shows that periodicity is a necessary condition for the occurrence of perfect state transfer in normal Cayley graphs, and  provides the minimum time at which perfect state transfer can occur.
	\begin{corollary}\label{minimum}
		Let $G$ be a normal Cayley graph. If $G$ exhibits perfect state transfer at time $\tau$, then $T_{\tau}(\mu)\in\{-1,1\}$ for every eigenvalue $\mu$ of $P$. Moreover, if $\tau$ is the minimum time for perfect state transfer, then $G$ is $2\tau$-periodic.
	\end{corollary}
	\begin{proof}
		Suppose $G$ exhibits perfect state transfer at time $\tau$ and let $\mu$ be an eigenvalue of $P$. Then by Lemma~\ref{pst_normalcayley}, $T_\tau(\mu)^2=1$, and so $T_{\tau}(\mu)\in\{-1,1\}$.
		
		Let $G$ be connected and  let $\tau$ be the minimum time for exhibiting perfect state transfer on $G$. Using properties of the Chebyshev polynomial of the first kind, it follows that $\mu=\cos \frac{m}{\tau}\pi$ for some positive integer $m$. Consequently,  Lemma~\ref{ev_grover} gives that $U^{2\tau}=I$.  Thus $G$ is periodic. Suppose that $G$ is $k$-periodic, that is, $k$ is the least positive integer such that $U^k=I$. This implies that $k\divides2\tau$. Since $\tau$ is the minimum time at which perfect state transfer occurs in $G$, we have $\tau\in\{1,\ldots,k-1\}$. Hence $k=2\tau$. 
		
		The preceding argument is true even if $G$ is disconnected. Indeed, if $G$ is disconnected, then it is a disjoint union of isomorphic connected Cayley graphs. Consequently, all connected components have the same spectrum and so the same period. Thus the conclusion remains valid.
	\end{proof} 
	Let $\Gamma$ be a finite group. Consider the complex vector space
	\[
	\mathbb{C}[\Gamma]
	:=
	\left\{
	\sum_{g\in\Gamma}c_g g
	:
	c_g\in\mathbb{C}
	\right\},
	\]
	where addition and scalar multiplication are  defined component-wise. The set $\Gamma$ forms a basis of the vector space $\mathbb{C}[\Gamma]$. The \emph{(left) regular representation} of $\Gamma$ is the homomorphism $L:\Gamma\to\operatorname{GL}(\mathbb{C}[\Gamma])$
	defined by
	\[
	L(g)\!\left(\sum_{h\in\Gamma}c_hh\right)
	=
	\sum_{h\in\Gamma}c_h(gh)
	\quad \text{for }g\in\Gamma.
	\]
	It is well known that (see, for instance, \cite[Theorem~4.4.4]{steinberg}) the regular representation decomposes as
	\begin{equation}\label{regular_rep}
		L
		\cong
		\bigoplus_{\rho\in\widehat{\Gamma}}d_\rho\,\rho.
	\end{equation}
	\begin{lemma}\label{pst_nes_normalcayley}
		Let $\Gamma$ be a finite group and let $z\in \Gamma\setminus\{e\}$. Then $z\in Z(\Gamma)$ is an element of order $2$  if and only if   $\chi_\rho(z)=\pm d_\rho$ for each $\rho\in \widehat{\Gamma}$.
	\end{lemma}	
	\begin{proof}
		Suppose that $z\in Z(\Gamma)$ is an element of order $2$. Then for each $\rho\in\widehat{\Gamma}$ and  $g\in\Gamma$,
		\[
		\rho(g)\rho(z)
		=
		\rho(gz)
		=
		\rho(zg)
		=
		\rho(z)\rho(g).
		\]
		Since $\rho$ is irreducible, Lemma~\ref{schur_lemma} implies that $\rho(z)=\lambda I_{d_\rho}$
		for some scalar $\lambda\in\mathbb{C}$. Using $z^2=e$, we obtain 
		$I_{d_\rho}=\rho(z^2)=\rho(z)^2=\lambda^2I_{d_\rho}$, 
		which yields $\lambda=\pm1$. Therefore $\rho(z)=\pm I_{d_\rho}$,
		and hence
		\[
		\chi_\rho(z)
		=
		\operatorname{tr}(\rho(z))
		=
		\pm\operatorname{tr}(I_{d_\rho})
		=
		\pm d_\rho.
		\]

		Conversely, let $\chi_\rho(z)=\pm d_\rho$ for each $\rho\in \widehat{\Gamma}$.  Let $\lambda_1,\ldots,\lambda_{d_\rho}$ be the eigenvalues of $\rho(z)$. Since $\rho(z)$ is unitary, $|\lambda_i|=1$ for $1\leq i\leq d_\rho$. Then
		\[
		d_\rho=|\chi_\rho(z)|
		=
		|\lambda_1+\cdots+\lambda_{d_\rho}|
		\le
		|\lambda_1|+\cdots+|\lambda_{d_\rho}|
		=d_\rho.
		\]
		Thus $	|\lambda_1+\cdots+\lambda_{d_\rho}|
		=
		|\lambda_1|+\cdots+|\lambda_{d_\rho}|
		=d_\rho$. This, along with $\chi_\rho(z)=\pm d_\rho$, give that $\lambda_1=\cdots=\lambda_{d_\rho}=1$ or $\lambda_1=\cdots=\lambda_{d_\rho}=-1$. 
		Hence $\rho(z)=\pm I_{d_\rho}$, giving that  $\rho(z^2)=I_{d_\rho}$.

		Consider the regular representation $L$ of $\Gamma$. Since $\rho(z^2)=I_{d_\rho}$ for every $\rho$, \eqref{regular_rep} gives $L(z^2)=I_{|\Gamma|}$. Note that the regular representation is faithful. Therefore $L(z^2)=I_{|\Gamma|}$ gives that $z^2=e$. Since $z\neq e$, the order of $z$ is $2$.  
		
		For $g\in \Gamma$ and $\rho\in\widehat{\Gamma}$, we have
		\[
		\rho(gzg^{-1})
		=
		\rho(g)\rho(z)\rho(g)^{-1}
		=
		\rho(g)(\pm I_{d_\rho})\rho(g)^{-1}
		=
		\pm I_{d_\rho}
		=
		\rho(z).
		\]
		Hence for each $\rho\in\widehat{\Gamma}$, $\rho(gag^{-1})=\rho(z)$. Therefore $L(gzg^{-1})=L(z)$. Since the regular representation is faithful, it follows that $gzg^{-1}=z$
		for each $g\in \Gamma$. Therefore $z\in Z(\Gamma)$.
	\end{proof}
	Define the subsets $\widehat{\Gamma}_{z,0}$ and $\widehat{\Gamma}_{z,1}$ of $\widehat{\Gamma}$ by
	\[\widehat{\Gamma}_{z,0}=\{\rho\in\widehat{\Gamma}:\chi_\rho(z)=d_\rho\}~~\text{and}~~\widehat{\Gamma}_{z,1}=\{\rho\in\widehat{\Gamma}:\chi_\rho(z)=-d_\rho\}.\]
	Thus if $z\in Z(\Gamma)$ is an element of order $2$, then $\widehat{\Gamma}=\widehat{\Gamma}_{z,0}\cup \widehat{\Gamma}_{z,1}$.

	Since the normal Cayley graph $\cay(\Gamma,S)$ is $|S|$-regular, it follows from Lemma~\ref{reg} and Theorem~\ref{ev_normal_cayley} that the eigenvalues of the discriminant matrix $P$ are given by
	\begin{equation}\label{evalue_p_normal}
		\mu_\rho=\frac{\lambda_\rho}{|S|} ~~\text{for}~~ \rho\in\widehat{\Gamma},
	\end{equation}
	where each $\mu_\rho$ has multiplicity $d_\rho^2$ and corresponding orthonormal eigenbasis
	\begin{equation}\label{evaector_p_normal}
		\{v_{i,j}^{\rho}:1\le i,j\le d_\rho\}.
	\end{equation}
	Therefore the  spectral decomposition of the discriminant $P$ of the normal Cayley graph $\cay(\Gamma,S)$ is 
	\begin{equation*}
		P=\sum_{\rho\in\widehat{\Gamma}}\mu_\rho\sum_{i=1}^{d_\rho}\sum_{j=1}^{d_\rho}v_{i,j}^\rho(v_{i,j}^\rho)^*.
	\end{equation*}
	Since $T_\tau(x)$ is a polynomial, it follows that
	\[T_\tau(P)=\sum_{\rho\in\widehat{\Gamma}}T_\tau(\mu_\rho)\sum_{i=1}^{d_\rho}\sum_{j=1}^{d_\rho}v_{i,j}^\rho(v_{i,j}^\rho)^*.\]
	The $(u,v)$-th entry of $T_\tau(P)$ is given by 
	\begin{align}
		T_\tau(P)_{u,v}&=\frac{1}{|\Gamma|}\sum_{\rho\in\widehat{\Gamma}}d_\rho T_\tau(\mu_\rho)\sum_{i=1}^{d_\rho}\sum_{j=1}^{d_\rho}\rho_{i,j}(u)\overline{\rho_{i,j}(v)} \nonumber\\
		&=\frac{1}{|\Gamma|}\sum_{\rho\in\widehat{\Gamma}}d_\rho T_\tau(\mu_\rho)\operatorname{Tr}(\rho(u)\rho(v)^{*})\nonumber\\
		&=\frac{1}{|\Gamma|}\sum_{\rho\in\widehat{\Gamma}}d_\rho T_\tau(\mu_\rho)\chi_\rho(uv^{-1})\label{uv_sd_normal_cayley}
	\end{align}
	\begin{theorem}\label{pst_nomralcayley_main}
		Let $G:=\cay(\Gamma,S)$ be a normal Cayley graph. Then $G$ exhibits perfect state transfer between the vertices $u$ and $v$ at time $\tau$ if and only if the following two conditions are satisfied.
		\begin{enumerate}
			\item[(i)] $z:=uv^{-1}$ is an element of order $2$ that lies in the center of $\Gamma$, and
			\item[(ii)] $T_\tau(\mu_\rho)=1$ for each $\rho\in\widehat{\Gamma}_{z,0}$ and  $T_\tau(\mu_\rho)=-1$ for each $\rho\in\widehat{\Gamma}_{z,1}$.
		\end{enumerate}
	\end{theorem}
	\begin{proof}
		Suppose the normal Cayley graph $\cay(\Gamma,S)$ exhibits perfect state transfer between $u$ and $v$ at time $\tau$. Then by Lemma~\ref{pst_normalcayley}, there exists $z\in Z(\Gamma)$ of order $2$ such that $(\mathcal{R}_z)_{u,v}=1$. This implies that $u=vz=zv$, and so $z=uv^{-1}$. Hence Condition $(i)$ holds.
		
		By Lemma~\ref{pst_nes_normalcayley}, $\chi_\rho(z)=\pm d_\rho$ for each $\rho\in \widehat{G}$. Therefore $\widehat{\Gamma}=\widehat{\Gamma}_{z,0}\cup \widehat{\Gamma}_{z,1}$. From \eqref{uv_sd_normal_cayley}, the $(u,v)$-th entry of $T_\tau(P)$ is given by
		\begin{equation}\label{eq1107}
			T_\tau(P)_{u,v}=\frac{1}{|\Gamma|}\sum_{\rho\in\widehat{\Gamma}}d_\rho T_\tau(\mu_\rho)\chi_\rho(z).
		\end{equation}
		Since $\sum_{\rho\in \widehat{\Gamma}}d_\rho^2=|\Gamma|$, Corollary~\ref{minimum} implies that $|T_\tau(P)_{u,v}|\leq 1$, where equality holds if and only if all the summands in \eqref{eq1107} have the same sign. Now let  $\rho_0$ be the trivial representation in $\widehat{\Gamma}$. Then $d_{\rho_0}=1$, $\chi_{\rho_0}(z)=1$, $\mu_{\rho_0}=1$, and hence $T_\tau(\mu_{\rho_0})=1$. Therefore, the common sign of the summands in \eqref{eq1107} must be positive. Consequently, $\chi_\rho(z)T_\tau(\mu_\rho)>0$ for each $\rho\in\widehat{\Gamma}$, which implies that $T_\tau(\mu_\rho)=1$ for each $\rho\in\widehat{\Gamma}_{z,0}$ and  $T_\tau(\mu_\rho)=-1$ for each $\rho\in\widehat{\Gamma}_{z,1}$. Thus Condition $(ii)$ holds.
		
		Conversely, assume that Conditions $(i)$ and $(ii)$ satisfy. Condition $(i)$ implies $\chi_\rho(z)=\pm d_\rho$ for each $\rho\in \widehat{\Gamma}$. Then using Condition $(ii)$ in \eqref{uv_sd_normal_cayley}, we obtain
		\[T_\tau(P)_{u,v}=\frac{1}{|\Gamma|}\sum_{\rho\in\widehat{\Gamma}}d_\rho^2=1.\]
		Therefore $\cay(\Gamma,S)$ exhibits perfect state transfer between $u$ and $v$ at time $\tau$.	
	\end{proof}
	\subsection{Abelian Cayley graphs}
	We now specialize Theorem~\ref{pst_nomralcayley_main} to Cayley graphs over finite abelian groups and obtain a characterization for the occurrence of perfect state transfer. Let $\Gamma$ be a finite abelian group with the identity element $0$. By the fundamental theorem of finite abelian groups,
	\[	\Gamma=\mathbb{Z}_{n_1}\oplus \cdots\oplus \mathbb{Z}_{n_k}\]
	for some integers $n_1,\ldots,n_k$ with $n_j\geq2$ for $1\leq j\leq k$.
	Thus every element $g\in\Gamma$ can be written uniquely as $g=(g_1,\ldots,g_{n_k})$, where $g_j\in\mathbb{Z}_{n_j}$ for $1\leq j\leq k$. For each $g\in\Gamma$, define the map $\chi_g:\Gamma\to\mathbb{C}$ by
	\[\chi_g(x)=\prod_{j=1}^{k}\omega_{n_j}^{g_jx_j}~~\text{for}~~x=(x_1,\ldots,x_k)\in\Gamma,\]
	where $\omega_{n}=\exp({2\pi \iu/n})$ . 
	
	Since $\Gamma$ is abelian, every irreducible representation of $\Gamma$ is one-dimensional. Hence we identify each irreducible representation with its corresponding character. In particular,
	\[
	\widehat{\Gamma}=\Irr(\Gamma)=\{\chi_g:g\in\Gamma\},
	\]
	and $d_\chi=1$ for every  $\chi\in\widehat{\Gamma}$. Then by Lemma~\ref{pst_nes_normalcayley}, $\chi_g(z)=\pm1$ for each $g\in\Gamma$ if and only if $z$ is an element of order $2$ in $\Gamma$. In that case, define
	\[\Gamma_{z,0}=\{g\in\Gamma:\chi_g(z)=1\}~~\text{and}~~\Gamma_{z,1}=\{g\in\Gamma:\chi_g(z)=-1\}.\]
	By \eqref{evalue_p_normal} and Theorem~\ref{ev_normal_cayley}, the eigenvalues of the discriminant $P$ of $\cay(\Gamma,S)$ are given by
	\begin{equation}\label{ev_p_cayley_abelian}
		\mu_g:=\mu_{\chi_g}=\frac{1}{d}\sum_{s\in S}\chi_g(s) ~~\text{for}~~g\in\Gamma,
	\end{equation}
	where $d=|S|$. The following result follows from Theorem~\ref{pst_nomralcayley_main}.
	\begin{theorem}\label{pst_abelian}
		Let $\Gamma$ be a finite abelian group. Then $\cay(\Gamma,S)$ exhibits perfect state transfer between $u$ and $v$ at time $\tau$ if and only if the following two conditions are satisfied.
		\begin{enumerate}
			\item[(i)] $z:=u-v$ is an element of order $2$, and
			\item[(ii)] $T_\tau(\mu_g)=1$ for each $g\in\Gamma_{z,0}$ and  $T_\tau(\mu_g)=-1$ for each $g\in\Gamma_{z,1}$.
		\end{enumerate}
	\end{theorem}
	Applying Theorem~\ref{pst_abelian} to the cyclic group $\Zl_n$, we obtain the following corollary. In this case, the unique element of order $2$ is $n/2$ (when $n$ is even), and $\Gamma_{n/2,0}$ and $\Gamma_{n/2,1}$ correspond to the even and odd numbers in $\Zl_n$, respectively.\pagebreak
	\begin{corollary}\label{pst_circulant}
		The circulant graph $\cay(\Zl_n,S)$ exhibits perfect state transfer between $u$ and $v$ at time $\tau$ if and only if the following two conditions are satisfied.
		\begin{enumerate}
			\item[(i)] $n$ is even with $u-v=n/2$, and
			\item[(ii)] $T_\tau(\mu_j)=(-1)^j$. 
		\end{enumerate}
	\end{corollary}
	We now use Theorem~\ref{pst_abelian} to construct an infinite family of abelian Cayley graphs exhibiting perfect state transfer. Let $\Gamma$ be a finite abelian group. For an element $z\in\Gamma$ of order $2$ and $g_0\in\Gamma\setminus\{0\}$, define $m_0=\lcm\{\operatorname{ord}(\chi_g(g_0)):g\in\Gamma_{z,0}\}$.
	\begin{lemma}\label{eg_pst_abelian}
		Let $z$ be an element of order $2$ of a finite abelian group $\Gamma$.  Let
		$S=\{\pm g_0,\pm(z+g_0)\}$ for some $g_0\in\Gamma\setminus\{0\}$.
		If $4\nmid m_0$, then $\cay(\Gamma,S)$ exhibits perfect state transfer at the minimum time $\lcm(4, m_0)/2$.
	\end{lemma}
	\begin{proof}
		Let $\widehat{\Gamma}=\{\chi_g:g\in\Gamma\}$. Since $z$ is an element of order $2$, we have $\chi_g(z)=\pm1$ for every $g\in\Gamma$. Hence by \eqref{ev_p_cayley_abelian}, the eigenvalues of the discriminant $P$ of $\cay(\Gamma,S)$ are given by
		\begin{align*}
			\mu_g
			&=\frac{1}{4}\bigl(\chi_g(g_0)+\chi_g(-g_0)+\chi_g(z+g_0)+\chi_g(-z-g_0)\bigr)\\
			&=\frac{1}{2}\Bigl(\operatorname{Re}(\chi_g(g_0))+\real(\chi_g(z)\chi_g(g_0))\Bigr)\\
			&= \begin{cases}
				\real(\chi_g(g_0)) & \text{if } g\in\Gamma_{z,0},\\
				0 & \text{if } g\in\Gamma_{z,1}.
			\end{cases}
		\end{align*}	
		
		Let  $4\nmid m_0$ and $\tau=\lcm(m_0,4)/2$. Note that $\chi_g(g_0)^{m_0}=1$ for each $g\in\Gamma_{z,0}$.  Since $m_0\mid\tau$, we have $\chi_g(g_0)^{\tau}=1$ for each $g\in\Gamma_{z,0}$. Hence for each $g\in\Gamma_{z,0}$, there exists an integer $r$ such that
		\[
		\chi_g(g_0)=\exp\!\left(\frac{2r\pi i }{\tau}\right).
		\]
		Thus $\real(\chi_g(g_0))=\cos (2r\pi /\tau)$. Therefore for every $g\in\Gamma_{z,0}$,
		\[
		T_\tau(\mu_g)
		=T_\tau\!\left(\cos\left(\frac{2r\pi }{\tau}\right)\right)
		=\cos(2r\pi)=1.
		\]
		Since $4\nmid m_0$, the integer $\tau/2$ is odd. Hence for every $g\in\Gamma_{z,1}$,
		\[
		T_\tau(\mu_g)
		=T_\tau(0)
		=T_\tau\!\left(\cos\frac{\pi}{2}\right)
		=\cos\!\left(\frac{\tau\pi}{2}\right)
		=-1.
		\]
		Thus by Theorem~\ref{pst_abelian}, $\cay(\Gamma,S)$ exhibits perfect state transfer at time $\tau$.

		By Corollary~\ref{minimum}, to determine the minimum time for exhibiting perfect state transfer, it suffices to determine the period of $G$. By Lemma~\ref{ev_grover}, the eigenvalues of the time evolution matrix $U$ of $\cay(\Gamma,S)$ are
		$\exp\!\left(\pm\iu\arccos\mu_g\right)$ for each $g\in\Gamma$, together with $\pm1$.
		If $g\in\Gamma_{z,0}$, then $\mu_g=\operatorname{Re}(\chi_g(g_0))$.
		Since $\chi_g(g_0)$ is a root of unity, $\exp\!\left(\pm\iu\arccos\mu_g\right)
		\in\{\chi_g(g_0),\chi_g(-g_0)\}$. Observe that $m_0$ is the least positive integer satisfying
		\[
		\exp\!\left(\pm\iu\arccos\mu_g\right)^{m_0}=1~\text{for each}~ g\in\Gamma_{z,0}.
		\]
		If $g\in\Gamma_{z,1}$, then $\mu_g=0$, and therefore
		$\exp\!\left(\pm\iu\arccos0\right)=\pm\iu$. Note that $4$ and $2$ are the least positive integers such that $(\pm\iu)^4=1$ and $(\pm1)^2=1$. Therefore by Lemma~\ref{periodic}, the period  of the graph is  $\lcm(4, m_0)$. Hence the minimum time for exhibiting perfect state transfer is $\lcm(4, m_0)/2$.
	\end{proof}
	Let $\Gamma$ be a finite abelian group. The exponent of $\Gamma$, denoted $e_\Gamma$, is the least positive integer such that $e_\Gamma a=0$ for every $a\in\Gamma$. By Theorem~\ref{pst_abelian}, if $e_\Gamma$ is odd, then $\cay(\Gamma,S)$ does not exhibit perfect state transfer for any connection set $S$. In contrast, suppose $e_\Gamma$ is even and $4\nmid e_\Gamma$. Since $\operatorname{ord}(\chi_g(g_0))\mid e_\Gamma$ for each $g\in\Gamma$, we have $m_0\mid 2\kappa$.  As $\kappa$ is odd, we have $4\nmid m_0$. Thus the next corollary follows directly from Lemma~\ref{eg_pst_abelian}.
	\begin{corollary}\label{eg_coro_pst_abelian}
		Let $\Gamma$ be a finite abelian group of exponent $e_\Gamma:=2\kappa$, where $\kappa$ is odd. Let $z\in\Gamma$ be an element of order $2$ and $g_0\in\Gamma\setminus\{0\}$. If $S=\{\pm g_0, \pm(z+g_0)\}$, then $\cay(\Gamma,S)$ exhibits perfect state transfer at the minimum time $\lcm(4, m_0)/2$.
	\end{corollary}
	Lemma~\ref{eg_pst_abelian} also provides a general construction of infinite families of circulant graphs exhibiting perfect state transfer as given in the following example. This example was initially given in \cite[Lemma~8.1]{kubota2}.
	\begin{example}\label{eg_circulant}
		Let $S=\{\pm a, \pm (\ell+a)\}\subset \Zl_{2\ell}\setminus\{0\}$.  If $4\nmid \ell$, then $\cay(\Zl_{2\ell},S)$ exhibits perfect state transfer at the minimum time \[\frac{1}{2}\lcm\left(4,\frac{\ell}{\gcd(a,\ell)}\right).\]
	\end{example}
	\begin{proof}
		Note that $\ell$ is an element of order $2$ in $\mathbb{Z}_{2\ell}$. Since $\chi_h(a)=\omega_{2\ell}^{ha}$ and $4\nmid\ell$, we have $4\nmid m_0$. Thus by Lemma~\ref{eg_pst_abelian}, $\cay(\Zl_{2\ell},S)$ exhibits perfect state transfer.
		
		Note that $m_0$ is the least positive integer satisfying $\omega_{2\ell}^{ham_0}=1$ for $1\leq h\leq 2\ell-1$ with $h$ even. This is equivalent to finding the least positive integer $m_0$ such that $\ell\mid am_0$. Hence $m_0=\ell/\gcd(a,\ell)$. Therefore  the minimum time for exhibiting perfect state transfer is obtained by Lemma~\ref{eg_pst_abelian}.
	\end{proof}
	The condition $4\nmid\ell$ cannot be omitted in the preceding example.  Recall that the discriminant eigenvalues of $\cay(\Zl_{2\ell},S)$ are
	\[
	\mu_h=
	\begin{cases}
		\cos\frac{ha}{\ell}\pi & \mbox{if }h \text{ is even} \\
		0 & \mbox{if }h \text{ is odd}.
	\end{cases}
	\]
	If $4\mid\ell$ and $a$ is odd, then   $\mu_{h}=0=\mu_{h-1}$ for $h=\ell/2$. Therefore by  Condition~$(ii)$ of Corollary~\ref{pst_circulant}, $\cay(\Gamma,S)$ does not exhibit perfect state transfer.
	\subsection{Dicyclic Cayley graphs}
	For $n\geq 2$, the \emph{dicyclic group} of order $4n$, denoted $T_{4n}$, is the group defined by the presentation
	\[
	T_{4n}=\langle a,b \mid a^{2n}=1,\; b^2=a^n,\; b^{-1}ab=a^{-1}\rangle.
	\]
	\begin{lemma}[{\cite[Exercises~17.6~and~18.3]{james_representations}}]
		Let $n$ be a positive integer. Then the irreducible representations of the dicyclic group $T_{4n}$ are given in Table~\ref{irr_t4n_odd} for odd $n$ and in Table~\ref{irr_t4n_even} for even $n$.
	\end{lemma}		
	\begin{table}[h]
		\centering
		\begin{tabular}{|c|c|c|}
			\hline
			& $a^{k}\ (0\le k\le 2n-1)$ & $ba^{k}\ (0\le k\le 2n-1)$\\
			\hline
			$\psi_1$
			&
			$1$
			&
			$1$
			\\
			$\psi_2$
			&
			$1$
			&
			$-1$
			\\
			$\psi_3$
			&
			$(-1)^k$
			&
			$\iu(-1)^k$
			\\
			$\psi_4$
			&
			$(-1)^k$
			&
			$\iu(-1)^{k+1}$
			\\
			$\rho_h\ (1\le h\le n-1)$
			&
			$\begin{pmatrix}
				\omega_{2n}^{hk} & 0\\
				0 & \omega_{2n}^{-hk}
			\end{pmatrix}$
			&
			$\begin{pmatrix}
				0 & \omega_{2n}^{-hk}\\
				\omega_{2n}^{h(k+n)} & 0
			\end{pmatrix}$
			\\
			\hline
		\end{tabular}
		\caption{Irreducible representations of $T_{4n}$ ($n$ odd).}
		\label{irr_t4n_odd}
	\end{table}
	\begin{table}[h]
		\centering
		\begin{tabular}{|c|c|c|}
			\hline
			& $a^{k}\ (0\le k\le 2n-1)$ & $ba^{k}\ (0\le k\le 2n-1)$\\
			\hline
			$\psi_1$
			&
			$1$
			&
			$1$
			\\
			$\psi_2$
			&
			$1$
			&
			$-1$
			\\
			$\psi_3$
			&
			$(-1)^k$
			&
			$(-1)^k$
			\\
			$\psi_4$
			&
			$(-1)^k$
			&
			$(-1)^{k+1}$
			\\
			$\rho_h\ (1\le h\le n-1)$
			&
			$\begin{pmatrix}
				\omega_{2n}^{hk} & 0\\
				0 & \omega_{2n}^{-hk}
			\end{pmatrix}$
			&
			$\begin{pmatrix}
				0 & \omega_{2n}^{-hk}\\
				\omega_{2n}^{h(k+n)} & 0
			\end{pmatrix}$
			\\
			\hline
		\end{tabular}
		\caption{Irreducible representations of $T_{4n}$ ($n$ even).}
		\label{irr_t4n_even}
	\end{table}
	
	Let $T_{4n}:=\{1,a,\ldots,a^{2n-1},b,ba,\ldots,ba^{2n-1}\}$ be the dicyclic group of order $4n$. Let $\cay(T_{4n},S)$ be a normal Cayley graph, where $1\notin S$ and $S=S^{-1}$. Write $S_a=S\cap\langle a\rangle$ and $S_b=(b^{-1}S)\cap\langle a\rangle$. We further partition $S_a$ and $S_b$ according to the parity of the exponent:
	\begin{align*}
		S_{a,0}&=\{a^j:1\le j\le 2n-1,\ a^j\in S_a,\ j\equiv0\pmod2\},\\
		S_{a,1}&=\{a^j:1\le j\le 2n-1,\ a^j\in S_a,\ j\equiv1\pmod2\},\\
		S_{b,0}&=\{a^j:0\le j\le 2n-1,\ a^j\in S_b,\ j\equiv0\pmod2\},~\text{and}\\
		S_{b,1}&=\{a^j:0\le j\le 2n-1,\ a^j\in S_b,\ j\equiv1\pmod2\}.
	\end{align*} 
	For convenience, denote $d_a=|S_a|$, $d_b=|S_b|$, $d_{a,k}=|S_{a,k}|$, $d_{b,k}=|S_{b,k}|$ for $k\in\{0,1\}$, and let  $d=|S|$.
	
	Since $S=S^{-1}$, if $n$ is odd, then the map $a^j\longmapsto a^{n+j}$
	defines a bijection between $S_{b,0}$ and $S_{b,1}$. Consequently, $d_{b,0}=d_{b,1}$.

	Using \eqref{evalue_p_normal} and Theorem~\ref{ev_normal_cayley}, we compute the eigenvalues of the discriminant matrix $P$ of $\cay(T_{4n},S)$. Let $\mu_j$ denote the discriminant eigenvalue corresponding to the one-dimensional irreducible representations $\psi_j$ for $1\le j\le 4$, and let $\widetilde{\mu}_h$ denote the discriminant eigenvalue corresponding to the two-dimensional irreducible representations $\rho_h$ for $1\le h\le n-1$. We first consider the one-dimensional irreducible representations.
	
	\noindent\textbf{Case 1. One-dimensional representations.}

	\noindent\textbf{Case 1.1.} $n$ is odd.
	\begin{align*}
		\mu_1&=1,\\
		\mu_2&=\frac{d_a-d_b}{d},\\
		\mu_3&=\mu_4=\frac{d_{a,0}-d_{a,1}}{d}.
	\end{align*}
	
	\noindent\textbf{Case 1.2.} $n$ is even.
	\begin{align*}
		\mu_1&=1,\\
		\mu_2&=\frac{d_a-d_b}{d},\\
		\mu_3&=\frac{d_{a,0}-d_{a,1}+d_{b,0}-d_{b,1}}{d},\\
		\mu_4&=\frac{d_{a,0}-d_{a,1}+d_{b,1}-d_{b,0}}{d}.
	\end{align*}
	
	\noindent\textbf{Case 2. Two-dimensional representations.}
	\[\widetilde\mu_h=\frac{1}{2d}\sum_{a^k\in S}\left(\omega_{2n}^{hk}+\omega_{2n}^{-hk}\right)=\frac{1}{d}\sum_{a^k\in S}\omega_{2n}^{hk}~~\text{for}~~1\leq h\leq n-1.\]
	Using the eigenvalues ${\mu}_j$ and $\widetilde{\mu}_h$, together with Theorem~\ref{pst_nomralcayley_main}, we obtain the following characterization of perfect state transfer on $\cay(T_{4n},S)$.
	\begin{theorem}\label{pst_dicyclic}
		Let $G:=\cay(T_{4n},S)$ be a normal Cayley graph. Then $G$ exhibits perfect state transfer between the vertices $u$ and $v$ at time $\tau$ if and only if $uv^{-1}=a^n$, and the eigenvalues of the discriminant $P$ satisfy
		\begin{enumerate}
			\item[(i)] $T_\tau(\mu_1)=T_\tau(\mu_2)=1$, $T_\tau(\mu_3)=T_\tau(\mu_4)=(-1)^n$, and
			\item[(ii)] $T_\tau(\widetilde\mu_h)=(-1)^h$ for $1\leq h\leq n-1$.
		\end{enumerate}
	\end{theorem}
	\begin{proof}
		Since $Z(T_{4n})=\{1,a^n\}$, the only central element of order $2$ is $a^n$. Moreover,
		\[
		\chi_{\psi_1}(a^n)=\chi_{\psi_2}(a^n)=1,\qquad
		\chi_{\psi_3}(a^n)=\chi_{\psi_4}(a^n)=(-1)^n,
		\]
		and
		\[
		\chi_{\rho_h}(a^n)=2(-1)^h~~\text{for}~~ 1\leq h\leq n-1.
		\]
		The result now follows immediately from Theorem~\ref{pst_nomralcayley_main}.
	\end{proof}
	We apply the previous result to the following example to illustrate its use.
	\begin{example}\label{eg_dicyclic}
		Let $T_{4n}$ be the dicyclic group with $4\nmid n$, and let $S=\{a^j,a^{n-j},a^{n+j},a^{2n-j}\}$ for some positive integer $j$. Then $\cay(T_{4n},S)$ exhibits perfect state transfer. Moreover, the minimum time for exhibiting perfect state transfer is 
		\[\frac{1}{2}\lcm\left(4,\frac{n}{\gcd(n,j)}\right).\]
	\end{example}
	\begin{proof}
		The eigenvalues of the discriminant $P$ of $\cay(T_{4n},S)$ are given by
		\begin{align*}
			{\mu}_1 &= 1= {\mu}_2,\\
			{\mu}_3 &= {\mu}_4=
			\begin{cases}
				0 & \text{if $n$ is odd}\\
				(-1)^j & \text{if $n$ is even},
			\end{cases}
		\end{align*}
		and for $1\leq h\leq n-1$,
		\[
		\widetilde\mu_h = \begin{cases}
			0 & \text{if $h$ is odd}\\
			\cos \frac{hj}{n}\pi & \text{if $h$ is even}.
		\end{cases}
		\]
		
		First, assume that $n$ is odd. Now using \eqref{chb}, it follows that 
		\begin{align*}
			T_{2n}(\mu_1)&=T_{2n}(\mu_2)=1,~~ T_{2n}(\mu_3)=T_{2n}(\mu_4)=-1 \\
			T_{2n}(\widetilde\mu_h)&=(-1)^h ~~\text{for}~~ 1\leq h\leq n-1.
		\end{align*}
		Thus by Theorem~\ref{pst_dicyclic},  $\cay(T_{4n},S)$ exhibits perfect state transfer at time $2n$.
		
		Now assume that $n$ is even. Since $4\nmid n$, we have $n\equiv2\pmod4$. Then 
		\begin{align*}
			T_{n}(\mu_1)&=T_{n}(\mu_2)=T_{n}(\mu_3)=T_{n}(\mu_4)=1 \\
			T_{n}(\widetilde\mu_h)&=(-1)^h ~~\text{for}~~ 1\leq h\leq n-1.
		\end{align*}
		Thus  $\cay(T_{4n},S)$ exhibits perfect state transfer at time $n$.

		Let $\tau$ be the  minimum time at which  $\cay(T_{4n},S)$ exhibits perfect state transfer. By Corollary~\ref{minimum}, $\cay(T_{4n},S)$ is $2\tau$-periodic. The discriminant  eigenvalues of $\cay(T_{4n},S)$ are $\pm1$, $0$, and $\cos \frac{hj}{n}\pi$ for $1\leq h\leq n-1$ with $h$ even. 
		The eigenvalues $1$, $-1$ and $0$ correspond to the  eigenvalues $1$, $-1$ and $\pm\iu$ of the time evolution matrix $U$. The eigenvalues $\cos \frac{hj}{n}\pi$, for $1\leq h\leq n-1$ with $h$ even, correspond to  the eigenvalues
		\[
		\eta_r:=\exp\left(\pm\frac{2rj\pi\iu}{n}\right),
		\qquad
		1\le r\le\left\lfloor\frac{n-1}{2}\right\rfloor,
		\]
		of the time evolution matrix.

		Let $k$ be the least positive integer such that $\eta_r^k=1$ for   $1\le r\le\left\lfloor\frac{n-1}{2}\right\rfloor$. This is equivalent to requiring	$n\mid rjk$ for $1\leq r\leq\lfloor(n-1)/2\rfloor$.	Consequently,
		\[
		k=\frac{n}{\gcd(n,j)}.
		\]
		Therefore by Lemma~\ref{periodic}, the period of $\cay(T_{4n},S)$ is  
		\[
		\lcm\left(4,\frac{n}{\gcd(n,j)}\right).
		\]
		Thus the result follows from Corollary~\ref{minimum}.		
	\end{proof}		
	
	\begin{corollary}\label{comb_dicyclic}
		Let $G:=\cay(T_{4n},S)$ be a normal Cayley for $n\geq 3$. If $G$ exhibits perfect state transfer, then one of the following holds.
		\begin{enumerate}
			\item[(i)] $(d_a,d_b)\in\{(n,n),~(n/3,n),~(2n/3,2n)\}$, or
			\item[(i)] $d_b=0$.
		\end{enumerate}
	\end{corollary}
	\begin{proof}
		Suppose $G$ exhibits perfect state transfer. Then by Corollary~\ref{minimum}, $G$ is periodic. Hence by Lemma~\ref{regular_periodic}, we have
		\[\mu_2=\frac{d_a-d_b}{d}\in\left\{\pm1,\pm\frac{1}{2},0\right\}.\]
		Since $d=d_a+d_b$, comparing each cases, we have $d_b=0$, $d_a=0$, $d_a=3d_b$, $d_b=3d_a$ or $d_a=d_b$. We now consider these cases separately.
		
		If $d_a=0$, then $\widetilde{\mu}_h=0$ for $1\leq h\leq n-1$. Consequently, Condition~$(ii)$ of Theorem~\ref{pst_dicyclic} implies that perfect state transfer does not occur in this case.
		
		Now suppose $d_a\neq 0$. Since $G$ is a normal Cayley graph, the connection set $S$ is a union of conjugacy classes of $T_{4n}$. The conjugacy classes of $T_{4n}$ are
		\[
		\{1\},~\{a^n\},~\{a^i,a^{-i}\}~\text{for}~1\leq i\leq n-1,~\{ba^{2j}:0\leq j\leq n-1\},~\{ba^{2j+1}:0\leq j\leq n-1\}.
		\]
		Since $1\notin S$ and $S$ is a union of conjugacy classes, we have $d_b\in\{0,n,2n\}$.
		Combining this with the above five cases, we obtain
		\begin{enumerate}
			\item[(i)] if $d_a=d_b$, then $(d_a,d_b)=(n,n)$;
			\item[(ii)] if $d_b=3d_a$, then $(d_a,d_b)=(n/3,n)$ or $(2n/3,2n)$;
			\item[(iii)] the remaining cases force either $d_b=0$ or $d_a=0$.
		\end{enumerate}
		Thus the result follows.
	\end{proof}
	
	Observe that the condition $d_b=0$ implies that the graph is disconnected, and an example of such a Cayley graph exhibiting perfect state transfer is given in Example~\ref{eg_dicyclic}. However,  we do not know any example satisfying the first condition of Corollary~\ref{comb_dicyclic} that exhibits perfect state transfer.
	\subsection{Dihedral Cayley graphs}
	For $n\geq2$, the \emph{dihedral group} of order $2n$, denoted $D_{2n}$, is the group defined by the presentation
	\[
	D_{2n}=\langle a,b \mid a^{n}=1,\; b^2=1,\; bab=a^{-1}\rangle.
	\]
	\begin{lemma}[{\cite[pp.~40]{rep_dihedral}}]
		Let $n$ be a positive integer. Then the irreducible representations of the dicyclic group $D_{2n}$ are given in Table~\ref{irr_dn_odd} for odd $n$ and in Table~\ref{irr_dn_even} for even $n$.
	\end{lemma}
	
	\begin{table}[h]
		\centering
		\begin{tabular}{|c|c|c|}
			\hline
			& $a^{k}\ (0 \le k \le n-1)$ & $ba^{k}\ (0 \le k \le n-1)$ \\
			\hline
			$\psi_{1}$ & $1$ & $1$ \\
			$\psi_{2}$ & $1$ & $-1$ \\
			$\rho_h \ \left(1 \le h \le \lfloor (n-1)/2 \rfloor\right)$
			&
			$\begin{pmatrix}
				\omega_n^{hk} & 0 \\
				0 & \omega_n^{-hk}
			\end{pmatrix}$
			&
			$\begin{pmatrix}
				0 & \omega_n^{-hk} \\
				\omega_n^{hk} & 0
			\end{pmatrix}$ \\
			\hline
		\end{tabular}
		\caption{Irreducible representations of $D_{2n}$ ($n$ odd).}
		\label{irr_dn_odd}
	\end{table}
	\begin{table}[h]
		\centering
		\begin{tabular}{|c|c|c|}
			\hline
			& $a^{k}\ (0 \le k \le n-1)$ & $ba^{k}\ (0 \le k \le n-1)$ \\
			\hline
			$\psi_{1}$ & $1$ & $1$ \\
			$\psi_{2}$ & $1$ & $-1$ \\
			$\psi_{3}$ & $(-1)^k$ & $(-1)^k$ \\
			$\psi_{4}$ & $(-1)^k$ & $(-1)^{k+1}$ \\
			$\rho_h \ \left(1 \le h \le \lfloor (n-1)/2 \rfloor\right)$
			&
			$\begin{pmatrix}
				\omega_n^{hk} & 0 \\
				0 & \omega_n^{-hk}
			\end{pmatrix}$
			&
			$\begin{pmatrix}
				0 & \omega_n^{-hk} \\
				\omega_n^{hk} & 0
			\end{pmatrix}$ \\
			\hline
		\end{tabular}
		\caption{Irreducible representations of $D_{2n}$ ($n$ even).}
		\label{irr_dn_even}
	\end{table}	
	Let $D_{2n}:=\{1,a,\ldots,a^{n-1},b,ba,\ldots,ba^{n-1}\}$ be the dihedral group of order $2n$, and consider the normal Cayley graph $\cay(D_{2n},S)$, where $1\notin S$ and $S=S^{-1}$. We adopt the same notation as for $\cay(T_{4n},S)$, namely, $S_a$, $S_b$, $S_{a,i}$, and $S_{b,i}$, together with their cardinalities $d_a$, $d_b$, $d_{a,i}$, and $d_{b,i}$, respectively, for $i\in\{0,1\}$.

	Using \eqref{evalue_p_normal} and Theorem~\ref{ev_normal_cayley}, we compute the eigenvalues of the discriminant matrix $P$ of $\cay(D_{2n},S)$. Let $\mu_j$ and $\widetilde{\mu}_h$ denote the discriminant eigenvalues corresponding to the one and two-dimensional irreducible representations, respectively.
	
	\pagebreak
	\noindent\textbf{Case 1. One-dimensional representations.}

	\noindent\textbf{Case 1.1.} $n$ is odd.
	\begin{align*}
		\mu_1&=1,\\
		\mu_2&=\frac{d_a-d_b}{d}.
	\end{align*}

	\noindent\textbf{Case 1.2.} $n$ is even.
	\begin{align*}
		\mu_1&=1,\\
		\mu_2&=\frac{d_a-d_b}{d},\\
		\mu_3&=\frac{d_{a,0}-d_{a,1}+d_{b,0}-d_{b,1}}{d},\\
		\mu_4&=\frac{d_{a,0}-d_{a,1}+d_{b,1}-d_{b,0}}{d}.
	\end{align*}

	\noindent\textbf{Case 2. Two-dimensional representations.}
	\[
	\widetilde{\mu}_h=
	\frac{1}{2d}\sum_{a^k\in S}\left(\omega_n^{hk}+\omega_n^{-hk}\right)= \frac{1}{d}\sum_{a^k\in S} \omega_n^{hk}~~\text{for}~~1\leq h\leq\left\lfloor\frac{n-1}{2}\right\rfloor.
	\]

	The center of $D_{2n}$ is trivial when $n$ is odd. Hence by Theorem~\ref{pst_nomralcayley_main}, a normal Cayley graph over $D_{2n}$ cannot exhibit perfect state transfer in this case. If $n$ is even, then $a^{n/2}$ is the unique element of order $2$ in the center of $D_{2n}$. The following result now follows from Theorem~\ref{pst_nomralcayley_main}, as in the proof of Theorem~\ref{pst_dicyclic}.
	\begin{theorem}\label{pst_dihedral}
		Let $G:=\cay(D_{2n},S)$ be a normal Cayley graph. Then $G$ exhibits perfect state transfer between the vertices $u$ and $v$ at time $\tau$ if and only if $n$ is even, $uv^{-1}=a^{n/2}$, and the eigenvalues of $P$ satisfies
		\begin{enumerate}
			\item[(i)] $T_\tau(\mu_1)=T_\tau(\mu_2)=1$, $T_\tau(\mu_3)=T_\tau(\mu_4)=(-1)^\frac{n}{2}$, and
			\item[(ii)] $T_\tau(\widetilde\mu_h)=(-1)^h$ for $1\leq h\leq\left\lfloor\frac{n-1}{2}\right\rfloor$.
		\end{enumerate}
	\end{theorem}
	The preceding result was established in \cite{bhakta6} using an explicit spectral analysis of $\cay(D_{2n},S)$. In contrast, it follows directly from our general characterization.
	\begin{example}\label{eg_dihedral}
		Let $D_{2n}$ be the dihedral group such that $n=2m$ with $4\nmid m$. Let $S=\{ a^j,a^{m-j},a^{m+j},a^{n-j}\}$ for some positive integer $j$. Then $\cay(D_{2n},S)$ exhibits perfect state transfer. Moreover, the minimum time for exhibiting perfect state transfer is 
		\[\frac{1}{2}\lcm\left(4,\frac{m}{\gcd(m,j)}\right).\]
	\end{example}
	\begin{proof}
		The eigenvalues of $P$ are given by
		\begin{align*}
			\widetilde{\mu}_1 &= 1= \widetilde{\mu}_2,\\
			\widetilde{\mu}_3 &= \widetilde{\mu}_4=
			\begin{cases}
				0 & \text{if $m$ is odd}\\
				(-1)^j & \text{if $m$ is even},
			\end{cases}
		\end{align*}
		and for $1\leq h\leq m-1$, 
		\[
		\mu_h = \begin{cases}
			0 & \text{if $h$ is odd}\\
			\cos \frac{hj}{m}\pi & \text{if $h$ is even}.
		\end{cases}
		\]
		The result follows from Theorem~\ref{pst_dihedral}, as in the proof of Example~\ref{eg_dicyclic}.
	\end{proof}
	The next result follows by an argument similar to that of Corollary~\ref{comb_dicyclic}.
	\begin{corollary}\label{comb_dihedral}
		Let $G:=\cay(D_{2n},S)$ be a normal Cayley graph. If $G$ exhibits perfect state transfer, then one of the following holds.
		\begin{enumerate}
			\item[(i)] $(d_a,d_b)\in\{(n/2,n/2),~(n/3,n),~(n/6,n/2)\}$, or
			\item[(i)] $d_b=0$.
		\end{enumerate}
	\end{corollary}

	As in the dicyclic case, the condition $d_b=0$ implies that the graph is disconnected, and Example~\ref{eg_dicyclic} provides such an example exhibiting perfect state transfer. Currently, we do not have any example satisfying the first condition of Corollary~\ref{comb_dihedral} that exhibits perfect state transfer.

	\section{Perfect state transfer on unitary Cayley graphs}\label{ucg}

	Let $\Zl_n$ be the additive group of integers modulo $n$, where $n\geq 2$ and let $U_n=\{a\in\Zl_n:\gcd(a,n)=1\}$. Then the Cayley graph $\cay(\Zl_n, U_n)$ is called the \emph{unitary Cayley graph}. We prefer to denote $\cay(\Zl_n,U_n)$ by $G_{\Zl_n}$. The graph $G_{\Zl_n}$ is a connected and $\varphi(n)$-regular graph, where $\varphi$ denotes Euler's totient function. Note that the unitary Cayley graph is a circulant graph. See~\cite{unitary_cayley} for more information about unitary Cayley graphs.
	
	The eigenvalues of the adjacency matrix of the circulant graph $\cay(\Zl_n,C)$ are given by  
	\begin{equation}\label{ev_circulant}
		\lambda_j=\sum_{s\in C} \omega_n^{js}~~\text{for}~0\leq j\leq n-1.
	\end{equation}
	In 1918, Ramanujan introduced a sum (now known as the Ramanujan sum) in his published seminar paper \cite{ramanujan}, defined by 
	\begin{equation}\label{ramanujan_sum}
		R(j,n)=\sum_{r\in U_n} \omega_n^{jr}=\mathop{\sum_{r\in U_n }}_{r<\frac{n}{2}} 2 \cos\left( \frac{2\pi jr }{n} \right) ~\text{for}~n,j\in \Zl,~ n\geq 1.
	\end{equation}
	The Ramanujan sum $R(j,n)$ can also be expressed in terms of arithmetic functions (see \cite{ramanujan}). For integers $n$, $j$ with $n\geq 1$, 
	$$ R(j,n)=\sum_{r\divides \gcd(j,n)} \mu\left(\frac{n}{r}\right)= \mu (c_{n,j})\frac{\varphi(n)}{\varphi(c_{n,j})},$$ 
	where $c_{n,j}=n/\gcd(n,j)$ and $\mu $ is the M\"{o}bius function.
	
	Therefore by \eqref{ev_circulant} and \eqref{ramanujan_sum}, the eigenvalues of the adjacency matrix of the unitary Cayley graph $G_{\Zl_n}$ are given by
	\begin{equation}\label{ev_unitarycayley}
		\lambda_j=R(j,n)=\mu (c_{n,j})\frac{\varphi(n)}{\varphi(c_{n,j})}~~\text{for}~~ 0\leq j\leq n-1.
	\end{equation} 
	Corollary~\ref{minimum} implies that if a graph exhibits perfect state transfer, then it is periodic. Hence we first characterize periodicity on $G_{\Zl_n}$.
	\begin{theorem}\label{period_unitary}
		The unitary Cayley graph $G_{\Zl_n}$ is periodic if and only if  $n=2^\alpha 3^\beta$, where $\alpha ~\text{and}~ \beta $ are nonnegative integers with $\alpha + \beta \neq 0 $. Moreover, the period $\tau$ of $G_{\mathbb Z_n}$ is given by
		\[
		\tau=\left\{ \begin{array}{rl}
			4 & \text{if } n=2^\alpha~(\alpha\geq1),\\
			12 & \text{if } n=2^\alpha 3^\beta ~(\beta\geq1).
		\end{array}\right.
		\]
	\end{theorem}
	\begin{proof}
		Since $G_{\mathbb Z_n}$ is integral and $\varphi(n)$-regular, it follows from Lemma~\ref{regular_periodic} that $G_{\mathbb Z_n}$ is periodic if and only if
		\begin{equation}\label{period_ucn}
			\spec_A(G_{\mathbb Z_n}) \subseteq \left\{ \pm \varphi(n), \pm \frac{\varphi(n)}{2}, 0\right\}.
		\end{equation}
		Suppose first that $G_{\mathbb Z_n}$ is periodic.  Let $p$ be a prime such that  $n=pm$ for some positive integer $m$. Note that $1\leq m \leq n-1$. We consider the eigenvalue $\lambda_{m}$ of the adjacency matrix of $G_{\mathbb Z_n}$. By \eqref{ev_unitarycayley}, we have $$\lambda_{m}=R(m,n)=\mu(p)\frac{\varphi(n)}{\varphi(p)}=-\frac{\varphi(n)}{\varphi(p)},~\text{as}~ c_{n,m}=p~\text{and}~\mu(p)=-1.$$ By  \eqref{period_ucn}, it is easy to conclude that $\varphi(p)=1$ or $2$, as $G_{\mathbb Z_n}$ is periodic. Therefore, the only possible values for the prime $p$ are $2$ and $3$. Hence 	if $G_{\Zl_n}$ is periodic, then the only prime factors of $n$ are $2$ or $3$. We now complete the proof in the following three cases. We use \eqref{ev_unitarycayley} to calculate the adjacency eigenvalues of $G_{\mathbb Z_n}$.

		\noindent\textbf{Case 1.} $n=2^\alpha$, where $\alpha\ge1$. The eigenvalues of the adjacency matrix of $G_{\mathbb{Z}_2}$ are $1$ and $-1$. Therefore $G_{\mathbb{Z}_2}$ is periodic.	 
		
		Now let $\alpha\geq2$. For $j=0$, $\lambda_0=\varphi(n)$. For $j=2^{\alpha -1},$ we have $\lambda_j=-\varphi(n)$.  Now let $j\notin\{ 0,\ 2^{\alpha-1}\}$ and $1\leq j\leq n-1$. If $\gcd(j,n)=1$ then $c_{n,j}=n$, and so $\lambda_j=R(j,n)=0$. If $\gcd(j,n)=2^r$ for some $r$ with $1\leq r\leq \alpha -2$, then $c_{n,j}=2^{\alpha-r}$. As $ \alpha - t\geq 2$, we have $\mu (c_{n,j})=0$. Hence $\lambda_j=0$. Thus we find that $\spec_A(G_{\mathbb Z_n})=\left\{\pm \varphi(n),\ 0\right\}$. Therefore $G_{\mathbb Z_n}$ is periodic for $n=2^\alpha$ as well, where $\alpha\geq 2$.

		\noindent\textbf{Case 2.} $n=3^\beta$, where $\beta\ge1$. The eigenvalues of the adjacency matrix of $G_{\mathbb{Z}_3}$ are $2$ and $-1$. Therefore $G_{\mathbb Z_3}$ is periodic. 
		
		Now let $\beta\geq 2$. For $j=0$, $\lambda_0=\varphi(n)$. For $j=3^{\beta -1}$,
		$\lambda_j=-\varphi(n)/2$. Next let $j\notin\{ 0,\ 3^{\beta-1}\}$ and $1\leq j\leq n-1$. If $\gcd(j,n)=1$ then $c_{n,j}=n$, and so 
		$\lambda_j=R(j,n)=0$. If $\gcd(j,n)=3^r$ for some $r$ with $1\leq r\leq \beta -1 $, then $c_{n,j}=3^{\beta-r}$. If $\beta - r\geq 2$ then $\mu (c_{n,j})=0$, and so $\lambda_j=0$. If $\beta - r= 1$, then the only possible value of $j$ is $2\times3^{\beta -1} $. In that case, $\ld_j = -\varphi(n)/2$. Therefore $$\spec_A(G_{\mathbb Z_n})=\left\{\varphi(n),\ -\frac{\varphi(n)}{2},\ 0\right\}.$$ Hence $G_{\mathbb Z_n}$ is also periodic for $n=3^\beta$, where $\beta\geq2$.

		\noindent\textbf{Case 3.} $n=2^\alpha3^\beta$, where $\alpha,\beta\ge1$. 	Let $\alpha =1$ and $\beta =1$. In this case, $n=6$ and $G_{\mathbb{Z}_6}$ has the adjacency spectrum $\{\pm 2, \pm 1\}$. Therefore $G_{\mathbb{Z}_6}$ is periodic. 
		
		Now let $n=2^\alpha  3^\beta$,  where $\alpha \geq 1 ~\text{and} ~ \beta \geq 1$ such that $\alpha \beta \neq 1$. For $j=0$, $\lambda_0=R(0,n)=\varphi(n)$. For $j=2^{\alpha} 3^{\beta -1}$, $\lambda_j=R(j,n)=-\varphi(n)/2$. For $j=2^{\alpha-1} 3^{\beta}$, $\lambda_j=R(j,n)=-\varphi(n)$. For $j=2^{\alpha -1}3^{\beta -1}$, $\lambda_j=R(j,n)=\varphi(n)/2$.
		Let $j\notin \{ 0,\ 2^{\alpha} 3^{\beta -1},\ 2^{\alpha-1} 3^{\beta},\ 2^{\alpha -1}3^{\beta -1}\}$ and $ 0\leq j \leq n-1$. If $\gcd(j,n)=1$ then $c_{n,j}=n$, and so $\lambda_j=R(j,n)=0$. Now let $\gcd(j,n)=2^s3^r$ for some nonnegative integers $s$ and $r$ such that  $0\leq s\leq \alpha$, $0\leq r\leq \beta $ and $s+r\neq 0$. Observe that either $s\leq \alpha -2$ or $r\leq \beta -2$, as $j\notin  \{0,\ 2^{\alpha} 3^{\beta -1},\ 2^{\alpha-1} 3^{\beta},\ 2^{\alpha -1}3^{\beta -1}\}$.  Therefore $c_{n,j}=2^{\alpha - s} 3^{\beta -r}$, where either $\alpha - s\geq 2~\text{or}~\beta -r \geq 2$. This gives $\mu (c_{n,j})=0$, and so $\lambda_j=0$. Thus we have 
		$$\spec_A(G_{\mathbb Z_n})=\left\{\pm \varphi(n),\ \pm \frac{\varphi(n)}{2},\ 0\right\}.$$
		Therefore $G_{\Zl_n}$ is periodic for $n=2^\alpha3^\beta$ as well, where $\alpha\geq 1$, $\beta\geq1$ and $\alpha\beta\neq 1$.
		
		For $n=2^\alpha~(\alpha\geq1)$, we have $\spec_A(G_{\mathbb Z_n})=\{\pm \varphi(n), 0\}$. As $G_{\mathbb Z_n}$ is $\varphi(n)$-regular, Lemma \ref{reg} gives that $\spec_P(G_{\mathbb Z_n})=\{\pm1,0\}$. Therefore by Lemma \ref{ev_grover}, $\spec_U(G_{\mathbb Z_n})=\{\pm \iu, \pm 1\}$. Hence by Corollary \ref{periodic}, the period of $G_{\mathbb Z_n}$ is $4$ for $n=2^\alpha$. Similarly, for $n=2^\alpha 3^\beta ~(\beta\geq1)$, the period of $G_{\mathbb Z_n}$ is $12$.  
	\end{proof}
	The preceding theorem yields an infinite family of periodic graphs. For example, the unitary Cayley graphs on 12 and 18 vertices, shown in Figure~\ref{fig}, are periodic.
	
	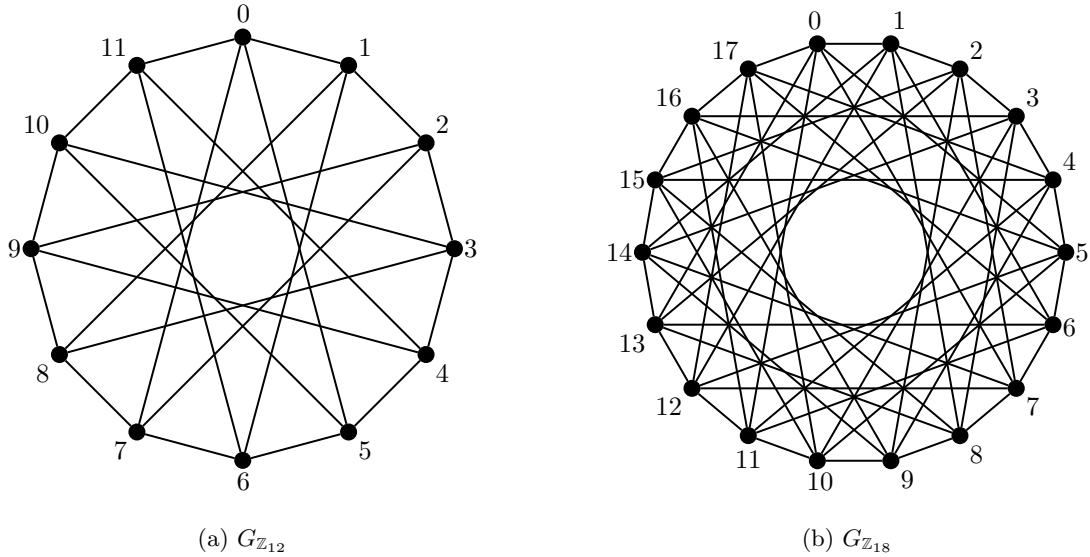
\begin{figure}[h!]
		\centering
		\begin{subfigure}{.5\textwidth}
			\centering

			\tikzset{every picture/.style={line width=0.75pt}} 
			
			\begin{tikzpicture}[scale=1.4] 
				
				\node[circle, draw, fill=black, inner sep=2pt] (Na) at ({360/12*0}:2) {};
				\node[circle, draw, fill=black, inner sep=2pt] (Nb) at ({360/12*1}:2) {};
				\node[circle, draw, fill=black, inner sep=2pt] (Nc) at ({360/12*2}:2) {};
				\node[circle, draw, fill=black, inner sep=2pt] (Nd) at ({360/12*3}:2) {};
				\node[circle, draw, fill=black, inner sep=2pt] (Ne) at ({360/12*4}:2) {};
				\node[circle, draw, fill=black, inner sep=2pt] (Nf) at ({360/12*5}:2) {};
				\node[circle, draw, fill=black, inner sep=2pt] (Ng) at ({360/12*6}:2) {};
				\node[circle, draw, fill=black, inner sep=2pt] (Nh) at ({360/12*7}:2) {};
				\node[circle, draw, fill=black, inner sep=2pt] (Ni) at ({360/12*8}:2) {};
				\node[circle, draw, fill=black, inner sep=2pt] (Nj) at ({360/12*9}:2) {};
				\node[circle, draw, fill=black, inner sep=2pt] (Nk) at ({360/12*10}:2) {};
				\node[circle, draw, fill=black, inner sep=2pt] (Nl) at ({360/12*11}:2) {};
				
				\node[right] at (Na) {$3$};
				\node[above right] at (Nb) {$2$};
				\node[above right] at (Nc) {$1$};
				\node[above right, xshift=-6pt, yshift=2pt] at (Nd) {$0$}; 
				\node[above left] at (Ne) {$11$};
				\node[above left] at (Nf) {$10$};
				\node[left] at (Ng) {$9$};
				\node[below left] at (Nh) {$8$};
				\node[below left] at (Ni) {$7$};
				\node[above right, xshift=-6pt, yshift=-15pt] at (Nj) {$6$}; 
				\node[below right] at (Nk) {$5$};
				\node[below right] at (Nl) {$4$};

				\draw (Na) -- (Nb);
				\draw (Nb) -- (Nc);
				\draw (Nc) -- (Nd);
				\draw (Nd) -- (Ne);
				\draw (Ne) -- (Nf);
				\draw (Nf) -- (Ng);
				\draw (Ng) -- (Nh);
				\draw (Nh) -- (Ni);
				\draw (Ni) -- (Nj);
				\draw (Nj) -- (Nk);
				\draw (Nk) -- (Nl);
				\draw (Nl) -- (Na);
				\draw (Nd) -- (Ni);
				\draw (Nd) -- (Nk);
				\draw (Nc) -- (Nj);
				\draw (Nc) -- (Nh);
				\draw (Nb) -- (Ni);
				\draw (Nb) -- (Ng);
				\draw (Na) -- (Nh);
				\draw (Na) -- (Nf);
				\draw (Nl) -- (Ng);
				\draw (Nl) -- (Ne);
				\draw (Nk) -- (Nf);
				\draw (Ne) -- (Nj);
				
			\end{tikzpicture}
			\caption{$G_{\Zl_{12}}$}\label{fig:sub1}
		\end{subfigure}%
		\begin{subfigure}{.5\textwidth}
			\centering

			\tikzset{every picture/.style={line width=0.75pt}} 
			
			\begin{tikzpicture}[scale=1.4]
				
				\node[circle, draw, fill=black, inner sep=2pt] (Na) at ({360/18*0}:2) {};
				\node[circle, draw, fill=black, inner sep=2pt] (Nb) at ({360/18*1}:2) {};
				\node[circle, draw, fill=black, inner sep=2pt] (Nc) at ({360/18*2}:2) {};
				\node[circle, draw, fill=black, inner sep=2pt] (Nd) at ({360/18*3}:2) {};
				\node[circle, draw, fill=black, inner sep=2pt] (Ne) at ({360/18*4}:2) {};
				\node[circle, draw, fill=black, inner sep=2pt] (Nf) at ({360/18*5}:2) {};
				\node[circle, draw, fill=black, inner sep=2pt] (Ng) at ({360/18*6}:2) {};
				\node[circle, draw, fill=black, inner sep=2pt] (Nh) at ({360/18*7}:2) {};
				\node[circle, draw, fill=black, inner sep=2pt] (Ni) at ({360/18*8}:2) {};
				\node[circle, draw, fill=black, inner sep=2pt] (Nj) at ({360/18*9}:2) {};
				\node[circle, draw, fill=black, inner sep=2pt] (Nk) at ({360/18*10}:2) {};
				\node[circle, draw, fill=black, inner sep=2pt] (Nl) at ({360/18*11}:2) {};
				\node[circle, draw, fill=black, inner sep=2pt] (Nm) at ({360/18*12}:2) {};
				\node[circle, draw, fill=black, inner sep=2pt] (Nn) at ({360/18*13}:2) {};
				\node[circle, draw, fill=black, inner sep=2pt] (No) at ({360/18*14}:2) {};
				\node[circle, draw, fill=black, inner sep=2pt] (Np) at ({360/18*15}:2) {};
				\node[circle, draw, fill=black, inner sep=2pt] (Nq) at ({360/18*16}:2) {};
				\node[circle, draw, fill=black, inner sep=2pt] (Nr) at ({360/18*17}:2) {};
				
				\node[right] at (Na) {$5$};
				\node[above right] at (Nb) {$4$};
				\node[above right] at (Nc) {$3$};
				\node[above right] at (Nd) {$2$};
				\node[above right, xshift=-3pt, yshift=2pt] at (Ne) {$1$};
				\node[above, xshift=-1pt, yshift=2pt] at (Nf) {$0$};
				\node[above left] at (Ng) {$17$};
				\node[above left] at (Nh) {$16$};
				\node[left] at (Ni) {$15$};
				\node[left] at (Nj) {$14$};
				\node[below left] at (Nk) {$13$};
				\node[below left] at (Nl) {$12$};
				\node[below, yshift=-2pt] at (Nm) {$11$};
				\node[below, xshift=1pt, yshift=-1pt] at (Nn) {$10$};
				\node[below right] at (No) {$9$};
				\node[below right] at (Np) {$8$};
				\node[below right] at (Nq) {$7$};
				\node[right, yshift=-1pt] at (Nr) {$6$};
				
				\draw (Na) -- (Nb);
				\draw (Nb) -- (Nc);
				\draw (Nc) -- (Nd);
				\draw (Nd) -- (Ne);
				\draw (Ne) -- (Nf);
				\draw (Nf) -- (Ng);
				\draw (Ng) -- (Nh);
				\draw (Nh) -- (Ni);
				\draw (Ni) -- (Nj);
				\draw (Nj) -- (Nk);
				\draw (Nk) -- (Nl);
				\draw (Nl) -- (Nm);
				\draw (Nm) -- (Nn);
				\draw (Nn) -- (No);
				\draw (No) -- (Np);
				\draw (Np) -- (Nq);
				\draw (Nq) -- (Nr);
				\draw (Nr) -- (Na);
				
				\draw (Na) -- (Nf);
				\draw (Nb) -- (Ng);
				\draw (Nc) -- (Nh);
				\draw (Nd) -- (Ni);
				\draw (Ne) -- (Nj);
				\draw (Nf) -- (Nk);
				\draw (Ng) -- (Nl);
				\draw (Nh) -- (Nm);
				\draw (Ni) -- (Nn);
				\draw (Nj) -- (No);
				\draw (Nk) -- (Np);
				\draw (Nl) -- (Nq);
				\draw (Nm) -- (Nr);
				
				\draw (Na) -- (Nh);
				\draw (Nb) -- (Ni);
				\draw (Nc) -- (Nj);
				\draw (Nd) -- (Nk);
				\draw (Ne) -- (Nl);
				\draw (Nf) -- (Nm);
				\draw (Ng) -- (Nn);
				\draw (Nh) -- (No);
				\draw (Ni) -- (Np);
				\draw (Nj) -- (Nq);
				\draw (Nk) -- (Nr);

				\draw (Nl) -- (Na);
				\draw (Nm) -- (Nb);
				\draw (Nn) -- (Nc);
				\draw (No) -- (Nd);
				\draw (Np) -- (Ne);
				\draw (Nq) -- (Nf);
				\draw (Nr) -- (Ng);
				
				\draw (Nn) -- (Na);
				\draw (No) -- (Nb);
				\draw (Np) -- (Nc);
				\draw (Nq) -- (Nd);
				\draw (Nr) -- (Ne);
			\end{tikzpicture}
			\caption{$G_{\Zl_{18}}$}
			\label{fig:sub2}
		\end{subfigure}
		\caption{Two examples of periodic unitary Cayley graphs}
		\label{fig}
	\end{figure}

	\begin{lemma}\label{cycle}
		The cycle $C_n$ exhibits perfect state transfer if and only if $n$ is even. Moreover, if $n$ is even, then $C_n$ exhibits perfect state transfer between antipodal vertices $u$ and $u+ n/2$ at the minimum time $n/2$.
	\end{lemma}
	\begin{proof}
		Since $C_n=\cay(\mathbb Z_n,\{\pm1\})$, Corollary~\ref{pst_circulant} implies that $C_n$ cannot exhibit perfect state transfer when $n$ is odd. Now let $n=2m$, where $m\ge2$. The eigenvalues of the discriminant matrix $P$ are
		\[
		\mu_j
		=
		\frac{1}{2}\left(\omega_n^j+\omega_n^{-j}\right)
		=
		\cos\left(\frac{2\pi j}{n}\right)
		=
		\cos\left(\frac{j\pi}{m}\right)~~\text{for}~~0\le j\le n-1.
		\]
		Hence by \eqref{chb},
		\[
		T_m(\mu_j)
		=
		\cos(j\pi)
		=
		(-1)^j~~\text{for}~~0\le j\le n-1.
		\] 	
		Therefore Corollary~\ref{pst_circulant} shows that $C_n$ exhibits perfect state transfer from the vertex $u$ to the antipodal vertex $u+n/2$ at time $n/2$.
		
		Finally, Kubota \emph{et al.}~\cite{mixed2} proved that $C_n$ is $n$-periodic. Hence by Corollary~\ref{minimum}, the minimum time at which perfect state transfer occurs is $n/2$.
	\end{proof}

	\begin{lemma}\label{complete}
		The complete graph $K_n$ exhibits perfect state transfer if and only if $n=2$.
	\end{lemma}
	\begin{proof}
		Since $\spec_P(K_n)=\{1,\frac{1}{1-n}\}$, it follows from Lemma~\ref{regular_periodic} that $K_n$ is periodic if and only if $n=2$ or $n=3$. Note that $K_n=\cay(\Zl_n, \{1,\hdots,n-1\})$. Suppose $K_n$ exhibits perfect state transfer. Since periodicity is a necessary condition for exhibiting  perfect state transfer in normal Cayley graphs, we find that the possible values of $n$ are either $2$ or $3$. Now by Corollary \ref{pst_circulant}, $K_2$ exhibits perfect state transfer, whereas $K_3$ does not. Hence the result follows.
	\end{proof}

	\begin{theorem}\label{pst_unitary}
		The unitary Cayley graph $G_{\mathbb Z_n}$ exhibits perfect state transfer if and only if $n\in\{2,4,6,12\}$.
	\end{theorem}
	\begin{proof}
		If perfect state transfer occurs in $G_{\mathbb Z_n}$, then by Corollary~\ref{minimum} and Theorem~\ref{period_unitary}, we have $n=2^\alpha 3^\beta$, where $\alpha ~\text{and}~ \beta $ are nonnegative integers with $\alpha + \beta \neq 0 $. For $n= 2,~ 4 $ and $6$, the graphs $G_{\mathbb Z_n}$ are  $K_2,~C_4$ and $C_6$, respectively. From Lemma \ref{complete} and Lemma \ref{cycle}, these graphs exhibit perfect state transfer.
		
		For $n=12$, the graph $G_{\mathbb Z_{12}}$ is $12$-periodic. Thus if $G_{\mathbb Z_n}$ exhibits perfect state transfer, then Corollary~\ref{minimum} implies that the minimum time for exhibiting perfect state transfer is $6$. The eigenvalues of the discriminant $P$ of $G_{\mathbb Z_{12}}$ are given by
		\[
		\mu_j=\frac{\ld_j}{\varphi(12)}=\left\{ \begin{array}{rl}
			1 & \mbox{if } j=0,\\
			-1 & \mbox{if } j=6,\\
			\frac{1}{2} & \mbox{if } j\in\{2,10\},\\
			-\frac{1}{2} & \mbox{if } j\in\{4,8\},\\
			0& \mbox{if }j\text{ odd.}
		\end{array}\right.
		\]
		Using \eqref{chb}, it follows that $T_6(\mu_j)=1$ if $j$ is even and $T_6(\mu_j)=-1$ if $j$ is odd.  Therefore by Corollary~\ref{pst_circulant}, $G_{\mathbb Z_{12}}$ exhibits perfect state transfer at the minimum time $6$.

		For $n=2^{\alpha}~ (\alpha \geq 3)$ or $n=2^{\alpha}3~ (\alpha\geq 3$) or $n=3^\beta 2~(\beta \geq 2)$ or $n=2^\alpha3^\beta~(\alpha\geq 2~\text{and}~\beta\geq 2)$, it follows from \eqref{ev_unitarycayley} that $\ld_1 =0=\ld_2$, and hence $\mu_1=\mu_2$. Therefore Condition~(ii) of Corollary~\ref{pst_circulant} is not satisfied, and consequently $G_{\mathbb{Z}_n}$ does not exhibit perfect state transfer. Finally, if $n=3^\beta$, then Condition~(i) of Corollary~\ref{pst_circulant} fails. Hence $G_{\mathbb{Z}_n}$ does not exhibit perfect state transfer in this case as well. This completes the proof.
	\end{proof}

	\subsection*{Acknowledgments}
	Koushik Bhakta acknowledges the support provided by the Prime Minister’s Research Fellowship (PMRF) scheme of the Government of India (PMRF-ID: 1903298). 
	

\end{document}